\documentclass[11pt,a4paper,onecolumn,accepted=2025-01-16]{quantumarticle}
\pdfoutput=1 
\usepackage[T1]{fontenc}

\usepackage{tikz}
\usetikzlibrary{positioning, calc, arrows}
\usepackage{graphicx}
\usepackage{amssymb}
\usepackage{physics}
\usepackage{array}
\usepackage{hyperref}
\usepackage[numbers,sort&compress]{natbib}
\usetikzlibrary{fadings}
\usetikzlibrary{decorations.pathmorphing}

\usepackage[capitalize,nameinlink]{cleveref}

\tikzset{snake it/.style={decorate, decoration=snake}}

\usetikzlibrary{decorations.pathreplacing,decorations.markings}
\usepackage{slashed}

\usepackage{float}
\usepackage{subcaption}

\DeclareMathOperator{\SMP}{\mathrm{SMP}}
\DeclareMathOperator{\sq}{\text{sq}}

\newcommand{\Disj}{\mathsf{Disj}}
\newcommand{\IPfunc}{\mathsf{IP}}

\tikzset{
    >=stealth',
    punkt/.style={
           rectangle,
           rounded corners,
           draw=black, very thick,
           text width=6.5em,
           minimum height=2em,
           text centered},
    pil/.style={
           ->,
           thick,
           shorten <=2pt,
           shorten >=2pt,},
  on each segment/.style={
    decorate,
    decoration={
      show path construction,
      moveto code={},
      lineto code={
        \path [#1]
        (\tikzinputsegmentfirst) -- (\tikzinputsegmentlast);
      },
      curveto code={
        \path [#1] (\tikzinputsegmentfirst)
        .. controls
        (\tikzinputsegmentsupporta) and (\tikzinputsegmentsupportb)
        ..
        (\tikzinputsegmentlast);
      },
      closepath code={
        \path [#1]
        (\tikzinputsegmentfirst) -- (\tikzinputsegmentlast);
      },
    },
  },
  mid arrow/.style={postaction={decorate,decoration={
        markings,
        mark=at position .5 with {\arrow[#1]{stealth'}}
      }}}
}

\newtheorem{theorem}{Theorem}

\newtheorem{corollary}[theorem]{Corollary}

\newtheorem{definition}[theorem]{Definition}

\newtheorem{lemma}[theorem]{Lemma}

\newenvironment{proof}[1][Proof]{\noindent\textbf{#1.}\,}{\ \rule{0.5em}{0.5em}}

\begin{document} 

\title{Linear gate bounds against natural functions for position-verification}

\author[1]{Vahid R. Asadi}
\email{vrasadi@uwaterloo.ca}
\orcid{0000-0001-8354-7463}

\author[1]{Richard Cleve}
\email{cleve@uwaterloo.ca}

\author[1]{Eric Culf}
\email{eculf@uwaterloo.ca}

\author[1,2]{Alex May}
\email{amay@perimeterinstitute.ca}
\orcid{0000-0002-4030-5410}

\affiliation[1]{Institute for Quantum Computing, University of Waterloo, Ontario, Canada}
\affiliation[2]{Perimeter Institute for Theoretical Physics, Waterloo, Ontario, Canada}

\begin{abstract}
A quantum position-verification scheme attempts to verify the spatial location of a prover. 
The prover is issued a challenge with quantum and classical inputs and must respond with appropriate timings. 
We consider two well-studied position-verification schemes known as $f$-routing and $f$-BB84. 
Both schemes require an honest prover to locally compute a classical function $f$ of inputs of length $n$, and manipulate $O(1)$ size quantum systems. 
We prove the number of quantum gates plus single qubit measurements needed to implement a function $f$ is lower bounded linearly by the communication complexity of $f$ in the simultaneous message passing model with shared entanglement. 
Taking $f(x,y)=\sum_i x_i y_i \,\,\text{mod} \,\,2$ to be the inner product function, we obtain an $\tilde{\Omega}(n)$ lower bound on quantum gates plus single qubit measurements. 
The scheme is feasible for a prover with linear classical resources and $O(1)$ quantum resources, and secure against sub-linear quantum resources. 
\end{abstract}

\vfill

\maketitle

\pagebreak

\tableofcontents

\section{Introduction}

The subject of position-verification considers how to establish the spatial location of a party or object, by interacting with them remotely. 
Verifying position may be a cryptographic goal in itself, or a building block used for other cryptographic constructions. 
As well, position-verification has recently been understood to be closely connected with other primitives in information-theoretic cryptography \cite{allerstorfer2023relating}, to topics in quantum gravity \cite{may2019quantum, may2020holographic,may2021holographic, may2022complexity}, to Hamiltonian simulation \cite{apel2024security}, and to unclonable secret sharing \cite{ananth2025unclonable}.

In a position-verification scheme, the verifier sends the prover quantum and classical systems and asks for a reply at a set of designated spacetime locations.
See \cref{fig:2dsetup} for a standard set-up in a spacetime with one spatial dimension. 
When the inputs and outputs are all classical there is no unconditionally secure verification scheme \cite{chandran2009position}. 
This is because the prover can intercept the input signals, copy and forward them, and compute the expected replies without ever entering the designated spacetime region. 
Since the no-cloning theorem precludes this copy and forward attack with quantum information, using quantum inputs was suggested as a potential route to secure position-verification \cite{kent2006tagging,kent2011quantum,malaney2016quantum}. 

\begin{figure}
    \centering
    \begin{subfigure}{0.45\textwidth}
    \begin{tikzpicture}[scale=0.6]
    
    \node[below left] at (-4,0) {$c_1$};
    \draw[fill=black] (-4,0) circle (0.15);

    \node[below right] at (4,0) {$c_2$};
    \draw[fill=black] (4,0) circle (0.15);

    \node[below right] at (4,8) {$r_2$};
    \draw[fill=blue] (4,8) circle (0.15);

    \node[below left] at (-4,8) {$r_1$};
    \draw[fill=blue] (-4,8) circle (0.15);
    
    \draw[fill=gray,opacity=0.5] (-1,1) -- (1,1) -- (1,7) -- (-1,7) -- (-1,1);

    \draw[->] (-5.5,1) -- (-5.5,2);
    \node[above] at (-5.5,2) {$t$};
    \draw[->] (-5.5,1) -- (-4.5,1);
    \node[right] at (-4.5,1) {$x$};
    
    \end{tikzpicture}
    \caption{}
    \label{fig:taggingsub1}
    \end{subfigure}
    \hfill
\begin{subfigure}{.45\textwidth}
\begin{tikzpicture}[scale=0.6]

    \node[below left] at (-4,0) {$c_1$};
    \draw[fill=black] (-4,0) circle (0.15);

    \node[below right] at (4,0) {$c_2$};
    \draw[fill=black] (4,0) circle (0.15);

    \node[below right] at (4,8) {$r_2$};
    \draw[fill=blue] (4,8) circle (0.15);

    \node[below left] at (-4,8) {$r_1$};
    \draw[fill=blue] (-4,8) circle (0.15);
    
    \draw[postaction={on each segment={mid arrow}}] (-4,0) -- (-2,2) -- (-2,6) -- (-4,8);
    \draw[postaction={on each segment={mid arrow}}] (4,0) -- (2,2) -- (2,6) -- (4,8);
    \draw[postaction={on each segment={mid arrow}}] (-2,2) -- (0,4) -- (2,6);
    \draw[postaction={on each segment={mid arrow}}] (2,2) -- (0,4) -- (-2,6);
    
    \draw[dashed] (2,2) -- (0,0) -- (-2,2);
    \node[below] at (0,0) {$\ket{\Psi}$};
    
    \draw[fill=yellow] (-2,2) circle (0.3);
    \draw[fill=yellow] (2,2) circle (0.3);
    \draw[fill=yellow] (-2,6) circle (0.3);
    \draw[fill=yellow] (2,6) circle (0.3);
    
    \draw[fill=gray,opacity=0.5] (-1,1) -- (1,1) -- (1,7) -- (-1,7) -- (-1,1);
    
\end{tikzpicture}
\caption{}
\label{fig:taggingsub2}
\end{subfigure}
    
\caption{A position-verification scheme in $1+1$ dimensions. Inputs are given at locations $c_1$, $c_2$. The prover should apply a designated quantum operation to these inputs, then return the outputs to points $r_1$, $r_2$. a) An honest prover enters the designated spacetime region (grey) to apply the needed quantum operation. b) A dishonest prover attempts to reproduce the same operation while acting outside the spacetime region. This leads to the definition of a non-local quantum computation. Figure reproduced from \cite{may2019quantum}.}
\label{fig:2dsetup}
\end{figure}
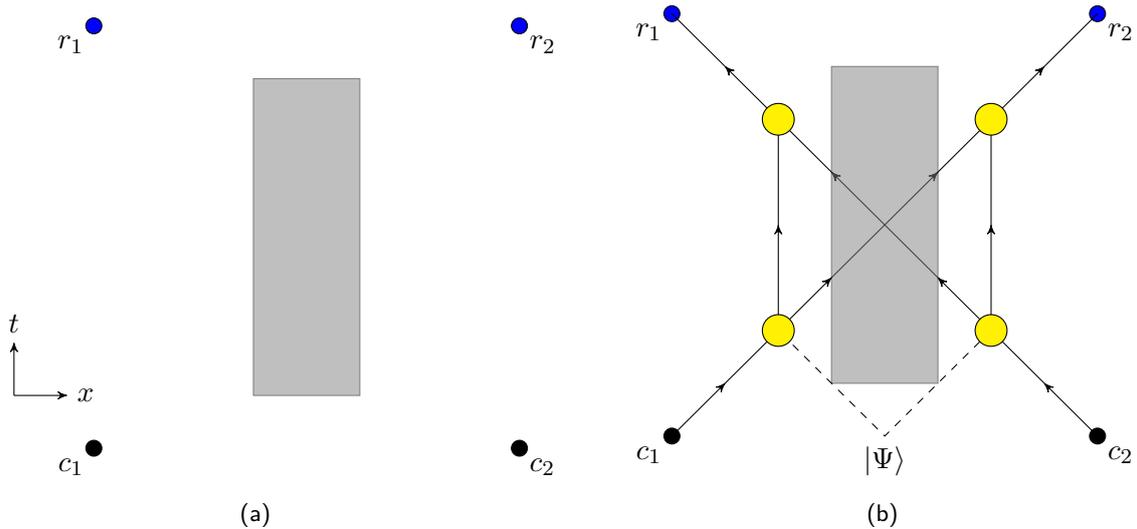

Even with quantum inputs, position-verification was proven insecure in the unconditional setting \cite{buhrman2014position,beigi2011simplified}.
The attacks use entanglement distributed across a spacetime region to simulate operations that might otherwise need to be implemented inside of the region. 
See \cref{fig:non-localandlocal}.
Following this no-go result, focus has shifted to proving security under the assumption of bounded entanglement or communication, with a number of works establishing lower bounds\footnote{Note that some of the existing lower bounds only bound the size of the resource system, rather than any measure of entanglement, or apply in the setting where the communication must be classical.}  \cite{beigi2011simplified,tomamichel2013monogamy,bluhm2021position, may2022complexity, gonzales2019bounds,cleve2026lower,may2026entanglement} and upper bounds \cite{beigi2011simplified,buhrman2013garden,cree2023code,gonzales2019bounds, speelman2015instantaneous, chakraborty2015practical} on entanglement requirements.
Another approach is to assume physical integrity of a device, which could contain a secret key \cite{kent2011quantum,cowperthwaite2023towards}.

In the bounded entanglement setting, particular attention has been paid to classes of protocols where most of the input is classical, with just $O(1)$ qubits, and in particular to schemes where an honest prover need only compute a classical function and do $O(1)$ quantum operations. 
In this context, it has been hoped that the quantum resource requirements would grow with the classical input size, so that a dishonest prover would need large quantum resources. 
Security of these schemes would then be based on an assumption that quantum resources are more difficult to prepare and implement than classical ones.

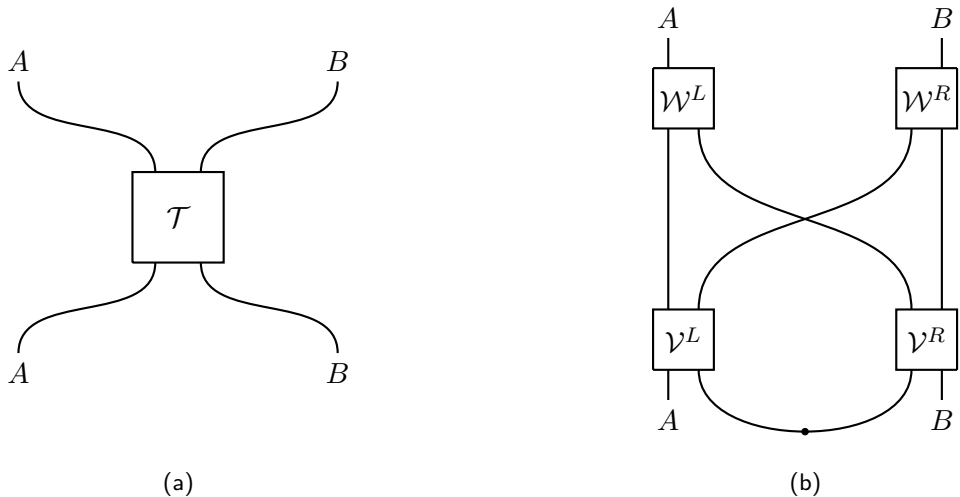
\begin{figure}
    \centering
    \begin{subfigure}{0.45\textwidth}
    \centering
    \begin{tikzpicture}[scale=0.6]
    
    \draw[thick] (-1,-1) -- (-1,1) -- (1,1) -- (1,-1) -- (-1,-1);
    
    \draw[thick] (-3.5,-3) to [out=90,in=-90] (-0.5,-1);
    \node[below] at (-3.5,-3) {$A$};
    \draw[thick] (3.5,-3) to [out=90,in=-90] (0.5,-1);
    \node[below] at (3.5,-3) {$B$};
    
    \draw[thick] (0.5,1) to [out=90,in=-90] (3.5,3);
    \node[above] at (3.5,3) {$B$};
    \draw[thick] (-0.5,1) to [out=90,in=-90] (-3.5,3);
    \node[above] at (-3.5,3) {$A$};
    
    \node at (0,0) {$\mathcal{T}$};

    \node at (0,-5) {$ $};
    
    \end{tikzpicture}
    \caption{}
    \label{fig:local}
    \end{subfigure}
    \hfill
    \begin{subfigure}{0.45\textwidth}
    \centering
    \begin{tikzpicture}[scale=0.4]
    
    \draw[thick] (-5,-5) -- (-5,-3) -- (-3,-3) -- (-3,-5) -- (-5,-5);
    \node at (-4,-4) {$\mathcal{V}^L$};
    
    \draw[thick] (5,-5) -- (5,-3) -- (3,-3) -- (3,-5) -- (5,-5);
    \node at (4,-4) {$\mathcal{V}^R$};
    
    \draw[thick] (5,5) -- (5,3) -- (3,3) -- (3,5) -- (5,5);
    \node at (4,4) {$\mathcal{W}^R$};
    
    \draw[thick] (-5,5) -- (-5,3) -- (-3,3) -- (-3,5) -- (-5,5);
    \node at (-4,4) {$\mathcal{W}^L$};
    
    \draw[thick] (-4.5,-3) -- (-4.5,3);
    
    \draw[thick] (4.5,-3) -- (4.5,3);
    
    \draw[thick] (-3.5,-3) to [out=90,in=-90] (3.5,3);
    
    \draw[thick] (3.5,-3) to [out=90,in=-90] (-3.5,3);
    
    \draw[thick] (-3.5,-5) to [out=-90,in=-90] (3.5,-5);
    \draw[black] plot [mark=*, mark size=3] coordinates{(0,-7.05)};
    
    \draw[thick] (-4.5,-6) -- (-4.5,-5);
    \node[below] at (-4.5,-6) {$A$};
    \draw[thick] (4.5,-6) -- (4.5,-5);
    \node[below] at (4.5,-6) {$B$};
    
    \draw[thick] (4.5,5) -- (4.5,6);
    \node[above] at (-4.5,6) {$A$};
    \draw[thick] (-4.5,5) -- (-4.5,6);
    \node[above] at (4.5,6) {$B$};
    
    \end{tikzpicture}
    \caption{}
    \label{fig:non-localcomputation}
    \end{subfigure}
    \caption{Local and non-local computations. a) A channel $\mathcal{T}_{AB\rightarrow AB}$ is implemented by directly interacting the input systems. b) A non-local quantum computation. The goal is for the action of this circuit on the $AB$ systems to approximate the channel $\mathcal{T}_{AB\rightarrow AB}$.}
    \label{fig:non-localandlocal}
\end{figure}

In a recent work, this hope was partly realized \cite{bluhm2021position}. 
The authors study two schemes referred to as $f$-routing and $f$-BB84.
In these schemes, an honest prover needs to compute a Boolean function $f$, then perform $O(1)$ quantum operations conditional on the value of $f$. 
For these tasks \cite{bluhm2021position} considers protocols that act unitarily on a shared resource system plus the inputs in the first round. 
In this setting, they prove that, with high probability over random choices of $f$, a dishonest prover needs to use a resource system composed of $q$ qubits with $q$ bounded below linearly in the number of classical input bits. 
Of the two variants for which this bound is proven, $f$-BB84 has the additional property that it can be made fully loss tolerant \cite{allerstorfer2023making}. 
This means that in experimental implementations, where most photons sent over long distances are lost, the $f$-BB84 protocol maintains a linear lower bound. 

While attractive, two important caveats remain in the practicality of these verification schemes and the applicability of this proof. 
First, as we discuss in more detail later on, the unitary view of \cite{bluhm2021position} means some resources that in a physical implementation can be classical are included in the system size they lower bound. 
Ideally, one would bound the quantum resources in the physical protocol, not the quantum resources in a purified (unitary) view of the protocol.

As well, the perspective of \cite{bluhm2021position} focuses on the distributed resources of the honest and dishonest prover but ignores the local computational resources. 
Considering this, we note that a random function, with overwhelming probability, is of exponential complexity. 
Thus in both $f$-routing and $f$-BB84 with random $f$, the honest prover needs exponential classical resources and $O(1)$ quantum resources, while the dishonest prover needs at least linear quantum resources and (at minimum) exponential classical resources (since they've also computed $f$). 
From this perspective, the honest prover's actions are not much easier than the dishonest one, and in any case a protocol requiring the computation of an exponential complexity function is not practical.\footnote{We can also compare this to \cite{liu2021beating}, which requires the honest prover to have a polynomial-time quantum computer and allows the use of only classical communication. From a computational perspective this is \emph{more} feasible than $f$-BB84 or $f$-routing scheme with a random choice of function, since BQP is weaker than EXP.} 
An interesting alternative given in \cite{bluhm2021position} is to use a low complexity function with a large communication complexity. 
For the inner product function, the authors prove a lower bound of $q = \Omega(\log n)$ with $n$ the classical input size.
While now the computation of the honest prover is linear complexity in the input size, it is still exponential complexity in the size of the quantum resource manipulated by the dishonest prover.

In this article, we focus on the computational requirements of the honest and dishonest players and give a new bound against $f$-routing and $f$-BB84.
Our bound follows from an extension of a lower bound strategy used in \cite{bluhm2021position}. 
We show that given a function $f$ a successful attack on $f$-BB84 or $f$-routing requires a dishonest player to implement a number of quantum gates linear in the simultaneous message passing with shared entanglement ($\SMP^*$) cost of $f$. 
As a concrete example, this places a linear (in the input size) lower bound against the number of quantum gates needed by a dishonest player to implement the inner product function.
Meanwhile, the honest player can implement $O(1)$ quantum gates and only linear classical gates. 
Our bound is also robust, applying for $f$-routing and $f$-BB84 protocols that have sufficiently small errors.\footnote{Specifically, the $f$-BB84 protocol should be correct with probability at least $0.945$ for every input, while the $f$-routing protocol needs output fidelity to be at least $0.9996$ for every input.} 

\vspace{0.2cm}
\noindent \textbf{Technical overview}
\vspace{0.2cm}

In simultaneous message passing ($\SMP$), two players receive inputs $x$ and $y$ and send messages to a referee that should determine $f(x,y)$. 
In $f$-routing ($f$-BB84 works similarly), a quantum system $Q$ is brought left or right based on the value of $f(x,y)$.
In relating these two settings, recall that a dishonest prover in $f$-routing has two agents, who intercept the input signals (see \cref{fig:non-localcomputation}). 
These become the two players in the $\SMP$ scenario. 
The key idea to reduce $f$-routing to $\SMP$ is to show that any successful attack on $f$-routing has a state after the first round operations that determines whether the input $Q$ is sent left or sent right. 
This means that this state determines the value of $f(x,y)$. 
The reduction works by having the players communicate data in an $\SMP$ protocol that is sufficient to reproduce this state to a referee. 
The referee can then determine if this is a state with $Q$ on the left or $Q$ on the right, and hence determine $f(x,y)$. 

Our main technical contribution is to adapt this reduction to lower bound the number of quantum gates and measurements applied by a dishonest prover, rather than to bound the dimensionality of their resource system as was done already in \cite{bluhm2021position}.  
As well, we work in an unpurified model, allowing the players to share a mixed state and apply general quantum operations.
Heuristically, a simple $f$-routing protocol would lead to a simple description of how to prepare this state, and hence to a good $\SMP$, so lower bounds on $\SMP$ lead to lower bounds on the complexity of the $f$-routing protocol. 
Importantly, we are interested in bounding the number of quantum operations performed by the prover, while allowing free classical processing. 
Direct use of the reduction from \cite{bluhm2021position} would instead bound the total number of classical and quantum operations. 
To obtain a bound on quantum operations alone requires that we modify the reduction to $\SMP$, and in particular reduce to simultaneous message passing with shared entanglement allowed between the players. 

The general form of the bound we obtain is stated below. 
In our bound, $q$ is the number of qubits held by each player, $C_M(f)$ is the number of measurements they jointly make, and $C_G(f)$ is the number of gates they jointly apply. 
We denote by $\SMP^*_{\epsilon',\delta}(f)$ the communication cost in the simultaneous message passing model with shared entanglement allowed, where we require $\epsilon'$ correctness on a fraction $1-\delta$ of the inputs. 
We have then
\begin{align}
    (\log(q)+1)(2C_G(f) +C_M(f)) \geq \SMP^*_{\epsilon',\delta}(f).
\end{align}
This is stated assuming the gate set is size $4$, but is easy to adapt to any gate set. 
Importantly, the bound holds even when allowing free classical processing, including the use of mid-circuit measurements and classical computations that make use of those mid-circuit measurement outcomes. 
 
We then exploit linear lower bounds on $\SMP^*$ for the inner product function given in \cite{cleve1998quantum} to prove an explicit lower bound,
\begin{align}
    (\log q + 1)(2C_G(\IPfunc)+ C_M(\IPfunc)) = \Omega(n).
\end{align}
Here $n$ is the number of input bits to the inner product function.
This bound on quantum gates and measurements holds even when allowing small errors in the $f$-routing or $f$-BB84 protocols, so long as they are sufficiently small.

\section{Background and tools}

\subsection{Distance measures and entropy inequalities}

In this section, we give a few definitions and collect some standard results for reference. 

Define the fidelity by
\begin{align}
    F(\rho,\sigma) = \left(\tr(\sqrt{\sqrt{\sigma}\rho\sqrt{\sigma}})\right)^2 ,
\end{align}
so that for pure states $F(\ket{\psi},\ket{\phi}) = |\braket{\psi}{\phi}|^2$. 
We also use the trace norm,
\begin{align}
    \Vert\rho -\sigma\Vert_1 \equiv \tr \sqrt{(\rho-\sigma)^\dagger(\rho-\sigma)}. 
\end{align}
The Fuchs-van de Graaf inequalities are
\begin{align}\label{eq:FVDG}
    1-\sqrt{F(\rho,\sigma)}\leq \frac{1}{2}\Vert \rho -\sigma\Vert_1 \leq \sqrt{1-F(\rho,\sigma)}.
\end{align}

We will make use of the complementary information trade-off (CIT) inequality \cite{renes2009conjectured}, which we state below. 
\begin{theorem}[\cite{renes2009conjectured}]\label{thm:CIT}
    Let ${\psi}_{REF}$ be an arbitrary tripartite state, with $R$ a single qubit. 
    We consider measurements on the $R$ system that produce a measurement result we store in a register $Z$. We consider measurements in both the computational and Hadamard bases, and denote the post-measurement state when measuring in the computational basis by $\rho_{ZEF}$, and when measuring in the Hadamard basis by $\sigma_{ZEF}$.  
    Then,
    \begin{align}
        H(Z|E)_\rho + H(Z|F)_\sigma \geq 1 .
    \end{align}
\end{theorem}

We also use Fano's inequality, given next. 
\begin{theorem}\textbf{(Fano's inequality \cite{fano1961transmission})}\label{thm:Fano} Let $X$ and $Y$ be random variables, and let $\hat{X}$ be a random variable describing a guess for the value of $X$ computed from a sample of $Y$. 
Then $p_{err}=\text{Pr}[\hat{X}\neq X]$ satisfies
\begin{align}
    h(p_{err})+p_{err}\log(|X|-1) \geq H(X|Y)
\end{align}
where $h(x)=-x\log x -(1-x)\log(1-x)$ is the binary entropy function.
\end{theorem}

Another useful statement is the continuity of the conditional entropy \cite{winter2016tight}. 

\begin{theorem}[\cite{winter2016tight}]\label{thm:condentropycont}
Suppose that
\begin{align}
    \frac{1}{2}\Vert\rho_{AB}-\sigma_{AB}\Vert_1 \leq \epsilon ,
\end{align}
and let $h(x)=-x\log x - (1-x)\log (1-x)$ be the binary entropy function. Then,
\begin{align}
    |H(A|B)_\rho - H(A|B)_\sigma| \leq 2\epsilon \log d_A + (1+\epsilon) h\left(\frac{\epsilon}{1+\epsilon}\right) ,
\end{align}
where $d_A$ is the dimension of the subsystem $A$.
\end{theorem}

Another object we will make use of is the squashed entanglement, 
\begin{align}
    E_{\sq}(A:B)_\rho \equiv \min_{\sigma_{ABC}:\tr_C\sigma=\rho} \frac{1}{2}I(A:B|C)_\sigma.
\end{align}
The squashed entanglement is faithful \cite{brandao2011faithful}, meaning it is zero if and only if it is evaluated on a separable state. 
For our purposes it is important that the squashed entanglement satisfies the following inequality \cite{koashi2004monogamy}, which expresses monogamy, 
\begin{align}\label{eq:squashedmonogamy}
    E_{\sq}(Q:A)_\sigma+E_{\sq}(Q:B)_\sigma\leq E_{\sq}(Q:AB)_\sigma.
\end{align}
Additionally, we need continuity of the squashed entanglement.
\begin{theorem}[\cite{li2018squashed}]\label{thm:squashedcontinuity} Suppose that
    \begin{align}
        \frac{1}{2}\Vert \rho-\sigma \Vert_1\leq \epsilon
    \end{align}
    and let $h(x) =-x \log x -(1-x) \log(1-x)$ be the binary entropy function. Then
    \begin{align}
    |E_{\sq}(A:B)_\sigma - E_{\sq}(A:B)_\rho| &\leq 4\epsilon \log d_A + 2(1+\epsilon)h\left(\frac{\epsilon}{1+\epsilon} \right),
    \end{align}
    where $d_A$ is the dimension of the subsystem $A$. 
\end{theorem}
We also use that the squashed entanglement satisfies the data processing inequality, 
\begin{align}\label{eq:squasheddataprocessing}
    E_{\sq}(A:B)_{\sigma_{AB}} \geq E_{\sq}(A:B)_{\mathcal{N}_A(\sigma_{AB})}.
\end{align}
Finally, we will use that the squashed entanglement is bounded above by the minimal log-dimension of its two inputs, $E_{\sq}(X:Y)\leq \min\{n_X,n_Y\}$.

\subsection{Communication complexity}

We will make use of a reduction from $f$-routing and $f$-BB84 to communication complexity scenarios. 
Specifically, we will be interested in the \emph{simultaneous message passing ($\SMP$)} scenario. 

A simultaneous message passing scenario is defined by a choice of function $f:\{0,1\}^n\times \{0,1\}^n\rightarrow \{0,1\}$. 
The scenario involves three parties, Alice, Bob, and the referee. 
Alice receives $x\in \{0,1\}^n$ and Bob receives $y\in \{0,1\}^n$. 
Alice and Bob compute messages $m_A, m_B$ from their local resources (including shared randomness) and the inputs they receive, and send their messages to the referee. 
Alice and Bob succeed if the referee can compute $f(x,y)$ from their messages. 
We define the $\SMP$ cost of $f$, denoted $\SMP(f)$ to be $\min_P(|m_A|+|m_B|)$ where the minimization is over choices of protocols.

There are several variations of the basic $\SMP$ scenario. 
For example, we can allow the referee to only succeed with some probability (taken over the shared randomness and selection of inputs,
we will always assume the input distribution is uniform in this work).
We denote the $\SMP$ cost in the case where they succeed with probability $1-\epsilon$ on at least $1-\delta$ fraction of inputs by $\SMP_{\epsilon,\delta}(f)$. 
We can also allow Alice and Bob to share entanglement as opposed to just classical randomness, and/or to send quantum messages. 
Our focus in this work is on the case where the messages are classical but they share entanglement. 
In this case, we denote the $\SMP$ cost by $\SMP^*_{\epsilon,\delta}(f)$. 

A formal definition follows.

\begin{definition}[$(\epsilon,\delta)$-$\SMP$ complexity]\label{def:SMP}
Let $f : \{0,1\}^n\times \{0,1\}^n\rightarrow \{0,1\}$ be a function, and $\epsilon, \delta \in [0,1]$ be parameters. An $\SMP$ protocol $P$ for $f$ consists of three algorithms Alice, Bob, and a referee. Alice receives $x\in \{0,1\}^n$ as input and outputs $m_A \in \{0,1\}^*$, Bob receives $y\in \{0,1\}^n$ as input and outputs $m_B \in \{0,1\}^*$, and the referee receives $m_A, m_B$ and outputs a bit $c=P(x,y)$. A protocol $P$ is $(\epsilon,\delta)$-correct if there exists $S \subseteq \{0,1\}^n\times \{0,1\}^n$ such that $\abs{S}\geq (1-\delta) \cdot 2^{2n}$, and
\begin{align}
    \forall (x,y) \in S : \Pr[P(x,y)=f(x,y)] \geq 1-\epsilon .
\end{align}
The $(\epsilon,\delta)$-$\SMP$ cost is defined by
\begin{equation*}
    \SMP_{\epsilon,\delta}(f) =  \min_{P: P \text{ is $(\epsilon,\delta)$-correct}}\left(\max_{(x,y)}(\abs{m_A} + \abs{m_B}) \right).
\end{equation*}
Similarly, we can define $\SMP^*_{\epsilon,\delta}(f)$ for the case where Alice and Bob share entanglement.
\end{definition}

Another basic communication complexity scenario allows communication back and forth between Alice and Bob. 
We can denote the minimal message size when using classical communication (we use the total number of bits sent by Alice and Bob) and allowing shared entanglement by $C^*_{\epsilon}(f)$ when the success probability of the protocol is at least $1-\epsilon$ (taken over the choice of inputs and internal randomness of the protocol).

An easy observation is that
\begin{align}\label{eq:smpandoneway}
    \SMP^*_{\epsilon,\delta}(f) \geq C^*_{\epsilon'}(f).
\end{align}
Here $\epsilon'=\delta+(1-\delta)\epsilon$.
The inequality follows because any $\SMP^*$ protocol can be turned into a two-way communication complexity scenario by having Alice send her message to Bob instead of the referee, and Bob run the same computation as the referee would in the $\SMP^*$ protocol.
The average error in this protocol is $\epsilon'=\delta+(1-\delta)\epsilon$ because the new protocol may be, in the worst case, wrong always on a fraction $\delta$ of inputs, and is wrong with probability at most $\epsilon$ on a fraction $1-\delta$ of inputs.

A standard function studied in communication complexity is the inner product, 
\begin{align}
    \IPfunc(x,y) = \sum_i x_i y_i \,\, \text{mod}\,\, 2 .
\end{align}
Intuitively, this is a difficult function to compute in communication complexity scenarios because the output depends sensitively on every bit of the input. 
More concretely, the following lower bound is proven in \cite{cleve1998quantum}. 
\begin{align}\label{eq:clevebound}
    C^*_\epsilon(\IPfunc) \geq \max \left\{\frac{1}{2}(1-2\epsilon)^2, (1-2\epsilon)^4\right\}n-1/2 .
\end{align}
We briefly comment on the proof of this theorem given in \cite{cleve1998quantum}. 
While not explicitly stated there, an inspection of their proof reveals their lower bound applies in the average-case setting.
In fact, it is only necessary for the protocol to work with probability $1-\epsilon$ over a uniform choice of input $x$ for at least one fixed choice of input $y$, or over a uniform choice of input $y$ at any fixed choice of $x$. 
A similar lower bound is shown in \cite{nayak2002communication} (Corollary 4.3), who find
\begin{align}\label{eq:IPCCbound}
    C^*_\epsilon(\IPfunc) \geq \frac{1}{2}n + 2\log(1-2\epsilon).
\end{align}
This also holds in the average-case (over inputs) setting.
Below in our explicit bounds, we apply this bound, rather than \cref{eq:clevebound}. 
Note however that we could also use \cref{eq:clevebound}, which is tighter for small values of the error parameter. 

As another example, it is shown in \cite{razborov2003quantum} that letting $\Disj : \{0,1\}^n \times \{0,1\}^n \to \{0,1\}$ be the disjointness function, we have whenever $\epsilon< 1/2$ that,
\begin{align}\label{eq:disjointnessbound}
    C^*_\epsilon(\Disj) = \Omega(\sqrt{n}).
\end{align}  

\section{Analysis of \texorpdfstring{$f$}{TEXT}-BB84}

\subsection{Definition and the strategy model}

We give the following definition of a \emph{qubit $f$-BB84} task.
In this definition and the rest of the text, we refer to the two agents of the prover (who sit on the left and right of the grey region in \cref{fig:taggingsub2}) as Alice and Bob. 
Because we are viewing cheating in the position-verification scenario as a form of a quantum game, which can be considered separately aside from the connection to position-verification, we also rename the role of the verifier as the referee. 

\begin{definition}\label{def:qubitfbb84}
    A \textbf{qubit} $f$\textbf{-BB84} task is defined by a choice of Boolean function $f:\{ 0,1\}^{n}\times\{0,1\}^n\rightarrow \{0,1\}$, and a $2$ dimensional Hilbert space $\mathcal{H}_Q$.
    Inputs $x\in \{0,1\}^{n}$ and quantum system $Q$ are given to Alice, and input $y\in \{0,1\}^{n}$ is given to Bob.
    The system $Q$ is in the maximally entangled state with a reference system $R$.
    Alice and Bob exchange one round of communication, with the combined systems received or kept by Alice labelled $M$ and the systems received or kept by Bob labelled $M'$.
    Define projectors
    \begin{align}
        \Pi^{q,b}=H^q\ketbra{b}{b}H^q .
    \end{align}
    The referee will measure $\{\Pi^{f(x,y),0},\Pi^{f(x,y),1}\}$ on the $R$ system and find measurement outcome $b\in \{0,1\}$. 
    The qubit $f$-BB84 task is completed $\epsilon$-correctly on input $(x,y)$ if Alice and Bob both output $b$. 
    More formally, Alice and Bob succeed $\epsilon$-correctly if they measure POVMs $\{\Lambda^{x,y,0}_{M},\Lambda^{x,y,1}_{M}\}$, $\{\Lambda^{x,y,0}_{M'},\Lambda^{x,y,1}_{M'}\}$ such that,
    \begin{align}
        \sum_b \tr(\Pi^{f(x,y),b}_R\otimes \Lambda^{x,y,b}_{M} \otimes \Lambda^{x,y,b}_{M'} \rho_{RMM'}) \geq 1-\epsilon .
    \end{align}
\end{definition}

Next, we give a fully general model capturing strategies that complete the $f$-BB84 task in the form of a non-local quantum computation.   

\begin{enumerate}
    \item Alice and Bob share a resource system ${\psi}_{AB}$. 
    \item The referee prepares $\Psi^+_{RQ}$ and hands $Q$ to Alice, preparing a joint state $\Psi^+_{RQ}\otimes \psi_{AB}$.
    \item At the same time as the above, Alice receives $x\in\{0,1\}^n$ and Bob receives $y\in\{0,1\}^n$.
    \item Alice applies $\mathcal{N}^x_{QA\rightarrow M_0M_0'}$, Bob applies $\mathcal{M}^y_{B\rightarrow M_1M_1'}$. Label $M=M_0M_1$, $M'=M_0'M_1'$ so that their joint state after the first round is 
    \begin{align}
        \rho_{RMM'}=\mathcal{N}^x_{QA\rightarrow M_0M_0'}\otimes \mathcal{M}^y_{B\rightarrow M_1M_1'}(\Psi^+_{RQ}\otimes \psi_{AB}) .
    \end{align} 
    \item $M_0$ and $M_1$ are sent to Alice, so that she holds $M$. At the same time, $M_0'$ and $M_1'$ are sent to Bob so that he holds $M'$. The (classical) inputs $x$ and $y$ are copied and sent to both parties. 
    \item Alice, Bob and the referee all compute $f(x,y)$. The referee measures $R$ in the $f(x,y)$ basis, obtaining measurement outcome $b$. Alice and Bob apply POVMs $\{\Lambda^{x,y,0}_{M},\Lambda^{x,y,1}_{M}\}$ and $\{\Lambda^{x,y,0}_{M'},\Lambda^{x,y,1}_{M'}\}$, then both output their measurement outcomes. 
\end{enumerate}
Recall that Alice and Bob succeed when they both obtain outcome $b$. 
See \cref{fig:fBB84andfR} for an illustration of a general protocol for $f$-BB84.

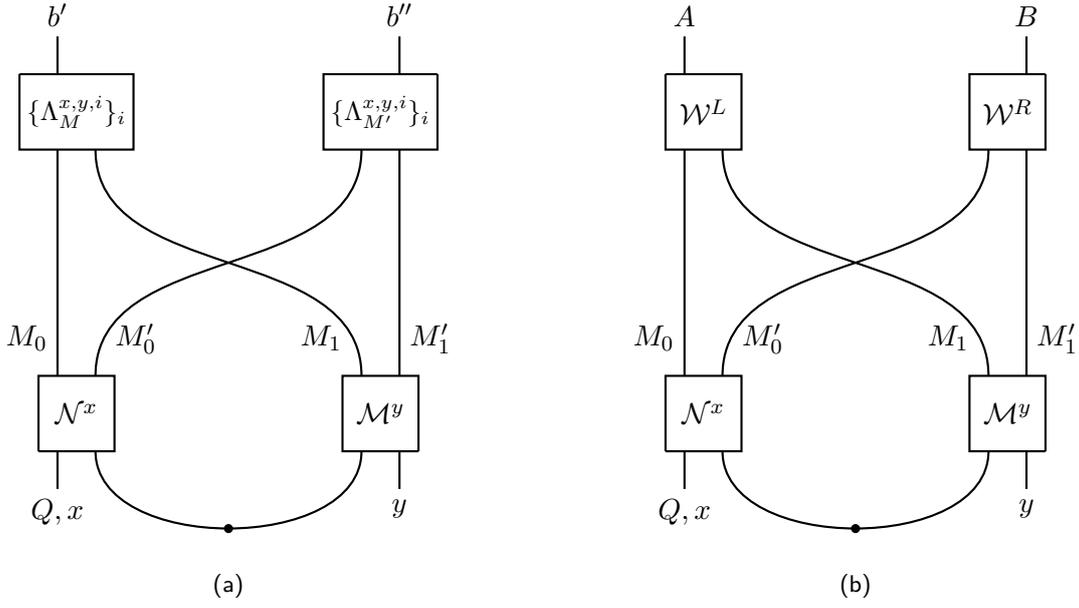
\begin{figure*}
    \centering
    \begin{subfigure}{0.45\textwidth}
    \centering
    \begin{tikzpicture}[scale=0.5]
    
    \draw[thick] (-5,-5) -- (-5,-3) -- (-3,-3) -- (-3,-5) -- (-5,-5);
    \node at (-4,-4) {$\mathcal{N}^x$};
    
    \draw[thick] (5,-5) -- (5,-3) -- (3,-3) -- (3,-5) -- (5,-5);
    \node at (4,-4) {$\mathcal{M}^y$};
    
    \draw[thick] (5.5,5) -- (5.5,3) -- (2.5,3) -- (2.5,5) -- (5.5,5);
    \node at (4,4) {\small{$\{\Lambda^{x,y,i}_{M'}\}_i$}};
    
    \draw[thick] (-5.5,5) -- (-5.5,3) -- (-2.5,3) -- (-2.5,5) -- (-5.5,5);
    \node at (-4,4) {\small{$\{\Lambda^{x,y,i}_M\}_i$}};
    
    \draw[thick] (-4.5,-3) -- (-4.5,3);
    \node[left] at (-4.5,-2) {$M_0$};
    
    \draw[thick] (4.5,-3) -- (4.5,3);
    \node[right] at (4.5,-2) {$M_1'$};
    
    \draw[thick] (-3.5,-3) to [out=90,in=-90] (3.5,3);
    \node[right] at (-3.25,-2) {$M_0'$};
    
    \draw[thick] (3.5,-3) to [out=90,in=-90] (-3.5,3);
    \node[left] at (3.25,-2) {$M_1$};
    
    \draw[thick] (-3.5,-5) to [out=-90,in=-90] (3.5,-5);
    \draw[black] plot [mark=*, mark size=3] coordinates{(0,-7.05)};
    
    \draw[thick] (-4.5,-6) -- (-4.5,-5);
    \node[below] at (-4.5,-6) {$Q,x$};
    \draw[thick] (4.5,-6) -- (4.5,-5);
    \node[below] at (4.5,-6) {$y$};
    
    \draw[thick] (4.5,5) -- (4.5,6);
    \node[above] at (-4.5,6) {$b'$};
    \draw[thick] (-4.5,5) -- (-4.5,6);
    \node[above] at (4.5,6) {$b''$};
    
    \end{tikzpicture}
    \caption{}
    \label{fig:fBB84}
    \end{subfigure}
    \hfill
    \begin{subfigure}{0.45\textwidth}
    \centering
    \begin{tikzpicture}[scale=0.5]
    
    \draw[thick] (-5,-5) -- (-5,-3) -- (-3,-3) -- (-3,-5) -- (-5,-5);
    \node at (-4,-4) {$\mathcal{N}^x$};
    
    \draw[thick] (5,-5) -- (5,-3) -- (3,-3) -- (3,-5) -- (5,-5);
    \node at (4,-4) {$\mathcal{M}^y$};
    
    \draw[thick] (5,5) -- (5,3) -- (3,3) -- (3,5) -- (5,5);
    \node at (4,4) {$\mathcal{W}^R$};
    
    \draw[thick] (-5,5) -- (-5,3) -- (-3,3) -- (-3,5) -- (-5,5);
    \node at (-4,4) {$\mathcal{W}^L$};
    
    \draw[thick] (-4.5,-3) -- (-4.5,3);
    \node[left] at (-4.5,-2) {$M_0$};
    
    \draw[thick] (4.5,-3) -- (4.5,3);
    \node[right] at (4.5,-2) {$M_1'$};
    
    \draw[thick] (-3.5,-3) to [out=90,in=-90] (3.5,3);
    \node[right] at (-3.25,-2) {$M_0'$};
    
    \draw[thick] (3.5,-3) to [out=90,in=-90] (-3.5,3);
    \node[left] at (3.25,-2) {$M_1$};
    
    \draw[thick] (-3.5,-5) to [out=-90,in=-90] (3.5,-5);
    \draw[black] plot [mark=*, mark size=3] coordinates{(0,-7.05)};
    
    \draw[thick] (-4.5,-6) -- (-4.5,-5);
    \node[below] at (-4.5,-6) {$Q,x$};
    \draw[thick] (4.5,-6) -- (4.5,-5);
    \node[below] at (4.5,-6) {$y$};
    
    \draw[thick] (4.5,5) -- (4.5,6);
    \node[above] at (-4.5,6) {$A$};
    \draw[thick] (-4.5,5) -- (-4.5,6);
    \node[above] at (4.5,6) {$B$};
    
    \end{tikzpicture}
    \caption{}
    \label{fig:fR}
    \end{subfigure}
    \caption{a) A general strategy for an $f$-BB84 scheme. Alice applies $\mathcal{N}^x$ in the first round, Bob applies $\mathcal{M}^y$. They communicate in one simultaneous exchange, and then apply measurements to their local systems. They succeed if $b=b'=b''$, with $b$ determined by measuring a reference maximally entangled with $Q$. b) A general strategy for an $f$-routing scheme. The first round operations are the same as before. In the second round, Alice and Bob apply channels mapping to qubit systems $A$, $B$. If $f(x,y)=0$, $A$ should be maximally entangled with the reference $R$. If $f(x,y)=1$, then $B$ should be maximally entangled with $R$.}
    \label{fig:fBB84andfR}
\end{figure*}

\subsection{Basis is determined in the first round}

In this section, we begin our analysis of the $f$-BB84 task. 
We show that the state after the first round of operations can only successfully complete the task for either $f(x,y)=0$ or $f(x,y)=1$, but not both. 
Thus, the state after the first round determines the basis of measurement $b$ performed by Alice and Bob.

We begin by defining the sets of states for which the protocol can succeed in $(x,y)\in f^{-1}(0)$ instances and a second set of states for which it can succeed in $(x,y)\in f^{-1}(1)$ instances. 

\begin{definition}\label{def:01setBB84}
    For $\epsilon \in [0,1/2]$, let $S^{\epsilon}_{0}$ be the set of states $\phi_{RMM'}$ such that there exists a measurement on subsystem $M$ and a measurement on subsystem $M'$ whose probability of both returning the outcome of a measurement in the computational basis on $R$ is at least $1-\epsilon$. 
    Similarly, let $S^{\epsilon}_{1}$ be the set of states $\phi_{RMM'}$ such that there exists a measurement on subsystem $M$ and a measurement on subsystem $M'$ whose probability of both returning the outcome of a measurement in the Hadamard basis on $R$ is at least $1-\epsilon$. 
\end{definition}

Next, we work towards proving that for small enough $\epsilon$, the sets $S_0^\epsilon$ and $S_1^\epsilon$ are disjoint. 
We begin with the following lemma.

\begin{lemma}\label{lemma:condentropybounds}
    Let $h(x)=-x\log x - (1-x) \log (1-x)$ be the binary entropy function.
    Suppose ${\psi^0}_{RMM'} \in S^\epsilon_{0}$, and $\rho^0_{ZMM'}$ is obtained by measuring $R$ in the computational basis, and storing the output in variable $Z$.  
    Then
    \begin{align}
        H(Z|M)_{\rho^0_{ZMM'}} &\leq h(\epsilon) , \nonumber \\
        H(Z|M')_{\rho^0_{ZMM'}} &\leq h(\epsilon) .
    \end{align}
    Further, suppose ${\psi^1}_{RMM'} \in S^\epsilon_{1}$, and $\sigma^1_{ZMM'}$ is obtained by measuring $R$ in the Hadamard basis and storing the output in variable $Z$.
    Then
    \begin{align}
        H(Z|M)_{\sigma^1_{ZMM'}} &\leq h(\epsilon) ,\nonumber \\
        H(Z|M')_{\sigma^1_{ZMM'}} &\leq h(\epsilon) .
    \end{align}
\end{lemma}
\begin{proof} We will prove that 
    \begin{align}
        H(Z|M)_{\rho^0_{ZMM'}} &\leq h(\epsilon) .
    \end{align}
    The remaining statements are similar. 
    Because ${\psi}_{RMM'}^0 \in S^\epsilon_{0}$, there exists a measurement on $M$ that determines the measurement outcome from measuring $R$ with probability $1-\epsilon$. 
    Let $W$ denote this measurement outcome from measuring $M$. 
    Then, Fano's inequality (\cref{thm:Fano}) gives that 
    \begin{align}
        \Pr(Z\neq W) \leq \epsilon \, \rightarrow H(Z|W)_{\rho^0}\leq h(\epsilon) .
    \end{align}
    By data processing inequality, we further obtain that
    \begin{align}
        H(Z|M)_{\rho^0} \leq H(Z|W)_{\rho^0} \leq h(\epsilon) ,
    \end{align}
    as needed. 
\end{proof}

Finally, we prove that the sets $S_0^\epsilon$ and $S_1^\epsilon$ are disjoint, as needed.\footnote{Note that this lemma and \cref{lemma:fRemptyintersection} are similar to lemmas appearing in \cite{bluhm2021position}, which in turn are similar to statements in \cite{buhrman2013garden}.}  

\begin{lemma}\label{lemma:fBB84emptyintersection}
    If $\epsilon\leq 0.11$, then we have $S^{\epsilon}_{0} \cap S^{\epsilon}_{1} = \emptyset$.
\end{lemma}
\begin{proof}
    Because $\psi^0_{RMM'} \in S^{\epsilon}_{0}$, we get from \cref{lemma:condentropybounds} that
    \begin{align}
        H(Z|M)_{\rho^0} \leq h(\epsilon) ,
    \end{align}
    where $\rho^0_{ZMM'}$ is obtained from ${\psi^0}_{RMM'}$ by measuring $R$ in the computational basis. 
    By \cref{thm:CIT} (CIT), this means
    \begin{align}
        H(Z|M')_{\sigma^0} \geq 1-h(\epsilon) . 
    \end{align}
    where $\sigma^0_{ZMM'}$ is obtained by measuring ${\psi^0}_{RMM'}$ in the Hadamard basis. 
    Now consider $\psi_{RMM'}^1\in S^{\epsilon}_{1}$. 
    By \cref{lemma:condentropybounds} this has
    \begin{align}
        H(Z|M')_{\sigma^1_{ZMM'}} \leq h(\epsilon) ,
    \end{align}
    with $\sigma^1$ obtained from $\psi^1_{RMM'}$ by measuring $R$ in the Hadamard basis. 
    But then
    \begin{align}
        H(Z|M')_{\sigma^0} - H(Z|M')_{\sigma^1} \geq 1-2h(\epsilon) .
    \end{align}
    Now we apply continuity of the conditional entropy (\cref{thm:condentropycont}) to upper bound this entropy difference in terms of the trace distance between the states $\sigma_{ZMM'}^0$, $\sigma^1_{ZMM'}$. 
    Defining $\Delta = \frac{1}{2}\Vert\sigma^0_{ZMM'}-\sigma^1_{ZMM'} \Vert_1$ and noting that $\Vert\sigma^0_{ZMM'}-\sigma^1_{ZMM'} \Vert_1 \geq \Vert\sigma^0_{ZM'}-\sigma^1_{ZM'} \Vert_1$, 
    \begin{align}
        1-2h(\epsilon)\leq H(Z|M')_{\sigma^0} - H(Z|M')_{\sigma^1} \leq 2\Delta + (1+\Delta) h\left(\frac{\Delta}{1+\Delta}\right) .
    \end{align}
    When $\epsilon<0.11$, $1-2h(\epsilon)>0$ and this places a non-trivial lower bound on $\Delta$. 
    But then
    \begin{align}
        2\Delta = \Vert\sigma^0_{ZMM'}-\sigma^1_{ZMM'} \Vert_1 \leq \Vert\psi^0_{RMM'}-\psi^1_{RMM'}\Vert_1
    \end{align}
    by monotonicity of the trace distance, so that states in $S_0^\epsilon$ and $S_1^\epsilon$ are separated by a non-zero distance, which proves the lemma. 
\end{proof}

The fact that $S_0^\epsilon \cap S_1^\epsilon=\emptyset$ will be used in our reduction from $f$-BB84 to $\SMP^*$. 
We give that reduction in \cref{sec:reduction}. 
Next, we prove a similar set separation for $f$-routing. 
This will allow us to use the same reduction for $f$-routing as well. 

\section{Analysis of \texorpdfstring{$f$}{TEXT}-routing}

\subsection{Definition and the strategy model}

We start by giving the following definition of a \emph{qubit $f$-routing} task. 

\begin{definition}\label{def:qubitfrouting}
    A \textbf{qubit $f$-routing} task is defined by a choice of Boolean function $f:\{0,1\}^{n}\times \{0,1\}^{n}\rightarrow \{0,1\}$, and a $2$ dimensional Hilbert space $\mathcal{H}_Q$.
    Inputs $x\in \{0,1\}^{n}$ and system $Q$ are given to Alice, and input $y\in \{0,1\}^{n}$ is given to Bob.
    Alice and Bob exchange one round of communication, with the combined systems received or kept by Alice labelled $M$ and the systems received or kept by Bob labelled $M'$.
    Label the combined actions of Alice and Bob in the first round as $\mathcal{N}^{x,y}_{Q\rightarrow MM'}$. 
    The qubit $f$-routing task is completed $\epsilon$-correctly on an input $(x,y)$ if Alice can recover $Q$ when $f(x,y)=0$ and Bob can recover $Q$ when $f(x,y)=1$, each with fidelity at least $1-\epsilon$. 
    More formally, the protocol is $\epsilon$-correct if there exists a channel $\mathcal{D}^{x,y}_{M\rightarrow Q}$ such that
    \begin{align}
        \text{when}\,\,\, f(x,y)=0,\,\,\,F\left(\mathcal{D}^{x,y}_{M\rightarrow Q} \circ\tr_{M'} \circ\mathcal{N}^{x,y}_{Q\rightarrow MM'}(\Psi^+_{RQ}),\Psi^+_{RQ}\right) \geq 1-\epsilon ,
    \end{align}
    and there exists a channel $\mathcal{D}^{x,y}_{M'\rightarrow Q}$ such that
    \begin{align}
         \text{when} \,\,\, f(x,y)=1,\,\,\, F\left(\mathcal{D}^{x,y}_{M'\rightarrow Q} \circ\tr_{M} \circ\mathcal{N}^{x,y}_{Q\rightarrow MM'}(\Psi^+_{RQ}),\Psi^+_{RQ}\right) \geq 1-\epsilon .
    \end{align}
\end{definition}

Next, we give a fully general model capturing strategies that complete the $f$-routing task in the form of a non-local quantum computation.

\begin{enumerate}
    \item Alice and Bob share a resource system $\psi_{AB}$ with $A$ held by Alice and $B$ held by Bob. 
    \item The referee prepares $\Psi^+_{RQ}$ and hands $Q$ to Alice.  
    \item At the same time as the above, Alice receives $x\in\{0,1\}^n$ and Bob receives $y\in\{0,1\}^n$.
    \item Alice applies $\mathcal{N}^x_{QA\rightarrow M_0M_0'}$, Bob applies $\mathcal{M}^y_{B\rightarrow M_1M_1'}$. Label $M=M_0M_1$, $M'=M_0'M_1'$ so that their joint state after the first round is 
    \begin{align}
        \rho_{RMM'}=\mathcal{N}^x_{QA\rightarrow M_0M_0'}\otimes \mathcal{M}^y_{B\rightarrow M_1M_1'}(\Psi^+_{RQ}\otimes \psi_{AB}) .
    \end{align} 
    \item $M_0$ and $M_1$ are sent to Alice, so that she holds $M$. At the same time, $M_0'$ and $M_1'$ are sent to Bob so that he holds $M'$. The inputs $x$ and $y$ are copied and sent to both parties. 
    \item Alice and Bob both compute $f(x,y)$. If $f(x,y)=0$, Alice applies a channel $\mathcal{D}^{x,y}_{M\rightarrow Q}$ and returns $Q$ to the referee. If $f(x,y)=1$, Bob applies a channel $\mathcal{D}^{x,y}_{M'\rightarrow Q}$ and returns $Q$ to the referee.
\end{enumerate}
See \cref{fig:fBB84andfR} for an illustration of this general strategy. 

\subsection{Routing is determined in the first round}

In this section, we adapt results from \cite{buhrman2013garden,bluhm2021position} that show the side on which a qubit maximally entangled with the reference can be returned to the referee is already determined after Alice and Bob apply their operations $\mathcal{N}^x, \mathcal{M}^y$. 
In other words, the state $\rho_{RMM'}$ determines where the qubit will be routed. 

To see this, we begin by defining sets of states for which the qubit can be produced on the left or right, respectively. 
\begin{definition}\label{def:01setrouting}
    We define the 0-set $\tilde{S}_0^{\epsilon}$ and 1-set $\tilde{S}_1^{\epsilon}$ as
    \begin{align}
        \tilde{S}_0^{\epsilon} &= \left\{\rho_{MM'R}: \exists \, \mathcal{N}_{M\rightarrow Q} \,s.t.\, F\left( \mathcal{N}_{M\rightarrow Q} \circ \tr_{M'} (\rho_{RMM'}),\Psi^+_{RQ}\right) \geq 1-\epsilon\right\} , \nonumber \\
        \tilde{S}_1^\epsilon &= \left\{\rho_{MM'R}: \exists \, \mathcal{N}_{M'\rightarrow Q} \,s.t.\, F\left(\mathcal{N}_{M'\rightarrow Q} \circ \tr_{M} (\rho_{RMM'}),\Psi^+_{RQ}\right) \geq 1- \epsilon  \right\} . \nonumber 
    \end{align}
\end{definition}
We would like to show that the sets $\tilde{S}_0^\epsilon$ and $\tilde{S}_1^\epsilon$ do not overlap when $\epsilon$ is suitably small.
Intuitively, the non-overlap of these sets indicates that the entanglement with $R$ has been brought to either Alice or Bob after the first round of operations -- if there is a way to recover the entanglement on the left then there is not one on the right, and vice versa. 
This can be understood as a consequence of the monogamy of entanglement.  

From here we can prove the following. 
\begin{lemma}\label{lemma:fRemptyintersection}
    If $\epsilon<\tilde{\epsilon}_0\approx 0.00085$, then $\tilde{S}_0^\epsilon \cap \tilde{S}_1^\epsilon =\emptyset$. 
\end{lemma}
\begin{proof}
    First consider $\psi^0\in \tilde{S}^\epsilon_0$. 
    Then by the definition of the set $\tilde{S}^\epsilon_0$ and the Fuchs-van de Graaf inequalities (\cref{eq:FVDG}), we have that
    \begin{align}\label{eq:sqrteps}
        \frac{1}{2}\Vert\mathcal{N}_{M\rightarrow Q}(\psi^0_{RM})-\Psi^+_{RQ} \Vert_1 \leq \sqrt{\epsilon}.
    \end{align}
    Then we have
    \begin{align}\label{eq:LBEsq}
        E_{\sq}(R:M)_{\psi^0} &\geq E_{\sq}(R:Q)_{\mathcal{N}_{M\rightarrow Q}(\psi^0)} \nonumber \\
        &\geq E_{\sq}(R:Q)_{\Psi^+} - 4\sqrt{\epsilon} \,n_{R} - g(\sqrt{\epsilon}) \nonumber \\
        &= (1-4\sqrt{\epsilon})n_{R} -g(\sqrt{\epsilon})
    \end{align}
    where the first inequality is data processing (\cref{eq:squasheddataprocessing}), the second comes from continuity of the squashed entanglement (\cref{thm:squashedcontinuity}) and \cref{eq:sqrteps}, and the last line is from evaluating the squashed entanglement on the maximally entangled state. 
    Note that we defined $g(x)=2(1+x)h\left(\frac{x}{1+x} \right)$.
    Next we apply monogamy of the squashed entanglement (\cref{eq:squashedmonogamy}), 
    \begin{align}
        E_{\sq}(R:M)_{\psi^0}+E_{\sq}(R:M')_{\psi^0} \leq E_{\sq}(R:MM')_{\psi^0} \leq n_{R}
    \end{align}
    where in the second inequality we used that the squashed entanglement is bounded above by the minimal log-dimension of its two inputs.
    Combined with the lower bound \eqref{eq:LBEsq}, the above gives
    \begin{align}
        E_{\sq}(R:M')_{\psi^0} \leq 4\sqrt{\epsilon}\, n_{R}+g(\sqrt{\epsilon}).
    \end{align}
    But now, we can also lower bound $E_{\sq}(R:M')_{\psi^1}$ using that $\psi^1\in \tilde{S}^\epsilon_1$, using the same sequence of steps as in \cref{eq:LBEsq}. 
    This gives that
    \begin{align}
        E_{\sq}(R:M')_{\psi^1} \geq (1-4\sqrt{\epsilon})n_{R} - g(\sqrt{\epsilon}).
    \end{align}
    If the upper bound on $E_{\sq}(R:M')_{\psi^0}$ is smaller than the lower bound on $E_{\sq}(R:M')_{\psi^1}$, then $E_{\sq}(R:M')_{\psi^1}\neq E_{\sq}(R:M')_{\psi^0}$, so that also $\psi^1\neq \psi^0$. 
    Comparing the upper and lower bounds, we have that  $\psi^1 \neq \psi^0$ whenever
    \begin{align}
        (1-4\sqrt{\epsilon})n_{R} - g(\sqrt{\epsilon}) > 4\sqrt{\epsilon}\, n_R+g(\sqrt{\epsilon}).
    \end{align}
    If for all $\psi^0\in \tilde{S}^\epsilon_0$, $\psi^1\in \tilde{S}^\epsilon_1$ we have $\psi^1\neq \psi^0$, then $\tilde{S}^\epsilon_0\cap \tilde{S}^\epsilon_1 = \emptyset$, as needed. 
    Thus $\tilde{S}^\epsilon_0\cap \tilde{S}^\epsilon_1 = \emptyset$ whenever the above inequality is satisfied. 
    We call the largest value such that the above inequality holds $\tilde{\epsilon}_0$. 
    Numerically, we find that for $n_R=1$, $\tilde{\epsilon}_0\approx 0.00085$. 
    As $n_R\rightarrow \infty$, $\tilde{\epsilon}_0$ approaches $1/64\approx 0.015$. 
\end{proof}

\section{Reduction to \texorpdfstring{$\SMP^*$}{TEXT} and lower bounds}\label{sec:reduction}

\subsection{Reduction to \texorpdfstring{$\SMP^*$}{TEXT}}

In this section, we consider lower bounds on the number of quantum gates Alice and Bob need to apply in order to successfully complete an $f$-BB84 or $f$-routing task. 
We show for certain functions such as the inner product function, this is linear (up to logarithmic factors) in the number of classical input bits $n$.

In more detail, we consider decomposing Alice and Bob's operations $\mathcal{N}^x$ and $\mathcal{M}^y$ into gates drawn from $\{T, H, S, CNOT\}$ acting on at most two qubits, and single qubit measurements in the computational basis. 
Since we want to bound Alice and Bob's quantum operations, we will allow them free classical processing. 
This classical processing could take as inputs $x,y$ and the outcomes from any mid-circuit measurements performed by Alice and Bob. 
In particular, the choice of gates later in the circuit can be conditioned on the outputs of classical processing involving earlier measurement outcomes. 
Notice that if we naively purify such a protocol, the classical processing which takes mid-circuit measurement outcomes as inputs will become a quantum operation. 
Thus bounding quantum operations in the purified view doesn't suffice to bound the quantum operations in the un-purified view, and hence doesn't bound the operations Alice and Bob are required to implement physically. 
Instead, we must directly bound the quantum operations in the un-purified view.

To do this, we first prove a reduction from $f$-BB84 or $f$-routing to $\SMP^*$. 

\begin{theorem}\label{thm:reduction}
    Suppose $P$ is an $f$-BB84 protocol that is $\epsilon<\epsilon_0=0.11$ correct, or an $f$-routing protocol that is $\epsilon<\tilde{\epsilon}_0\approx 0.00085$ correct, on a $1-\delta$ fraction of the inputs, uses $C_G(f)$ gates drawn from a gate set of size $4$ and also uses $C_M(f)$ single qubit measurements in the computational basis.
    Then, 
    \begin{align}\label{eq:fRgatelowerbound}
        (\log(q)+1)(2C_G(f) +C_M(f)) \geq \SMP^*_{\epsilon',\delta}(f) ,
    \end{align}
    where $q$ is the number of qubits held by Alice and Bob, and $\SMP^{*}_{\epsilon',\delta}(f)$ denotes the minimal message size needed to compute $f(x,y)$ in the $\SMP^*$ model with correctness $\epsilon'=\epsilon/\epsilon_0$ for $f$-BB84 and $\epsilon'=\epsilon/\tilde{\epsilon}_0$ for $f$-routing on at least $1-\delta$ fraction of possible inputs. 
\end{theorem}
\begin{proof}
We consider an $f$-BB84 protocol and show it defines an $\SMP^*$ protocol. 
The referee holds a classical description of the initial resource state. 
Alice and Bob share the resource system. 
Alice and Bob's strategy will be to send the referee a description of their local operations. 
We consider a decomposition of Alice and Bob's operations into gates and measurements. 
Alice and Bob apply their operations to their shared resource state and the input system. 
As they do so, they keep a record of the gates they apply (which may be computed using mid-circuit measurement outcomes) and their measurement outcomes $m$, then send this to the referee. 
The referee will then compute a classical description of the state $\rho_{RMM'}(m)$ and determine if it is inside of ${S}_0^{\epsilon_0}$ or ${S}_1^{\epsilon_0}$. 
By \cref{lemma:fBB84emptyintersection}, these sets are disjoint so this procedure is unambiguous.
We show below that, as a consequence of correctness of the $f$-BB84 protocol, with high probability $\rho_{RMM'}(m)$ is inside the set ${S}_{f(x,y)}^{\epsilon_0}$, so that the $\SMP^*$ protocol is correct with high probability. 

To determine how much communication this protocol uses, consider that for each gate they specify the gate choice, requiring $2$ bits, and the location of the gate, which requires $2\log q$ bits for a contribution of $(2\log q + 2) C_G(f)$ bits. 
Further, to specify each measurement requires $\log q$ bits to specify where the measurement occurs plus $1$ bit to specify the measurement outcome, for a contribution of $(\log q+1)C_M(f)$. 
The total message size sent by Alice and Bob then is the left hand side of \cref{eq:fRgatelowerbound}. 

It remains to show that $\rho_{RMM'}(m)$ is inside of ${S}_{f(x,y)}^{\epsilon_0}$ with high probability over the measurement outcomes $m$. 
We first establish this for a pair of inputs $(x,y)\in f^{-1}(0)$ which is $\epsilon$-correct; an input $(x,y)\in f^{-1}(1)$ can be handled similarly. 
By correctness of the $f$-BB84 protocol, we have that
\begin{align}
        \sum_b \tr(\Pi^{f(x,y),b}_R\otimes \Lambda^{x,y,b}_{M} \otimes \Lambda^{x,y,b}_{M'} \rho_{RMM'}) \geq 1-\epsilon .
\end{align}
Using that $\rho_{RMM'}=\sum_m p_m\rho_{RMM'}(m)$, this says
\begin{align}
    \sum_m p_m \sum_b \tr(\Pi^{f(x,y),b}_R\otimes \Lambda^{x,y,b}_{M} \otimes \Lambda^{x,y,b}_{M'} \rho_{RMM'}(m)) \geq 1-\epsilon .
\end{align}
Define the random variable $P_m=1-\sum_b \tr(\Pi^{f(x,y),b}_R\otimes \Lambda^{x,y,b}_{M} \otimes \Lambda^{x,y,b}_{M'} \rho_{RMM'}(m))$, so that the above reads $\langle P_m \rangle \leq \epsilon$. 
So long as $P_m \leq {\epsilon}_0$ we will have that $\rho_{RMM'}(m)\in {S}^{\epsilon_0}_{0}$, so the referee fails only when $P_m > {\epsilon}_0$. 
By Markov's inequality, this occurs with probability
\begin{align}
    \Pr[P_m > \epsilon_0] \leq \frac{\epsilon}{\epsilon_0}. 
\end{align}
Thus the referee succeeds with probability $p\geq 1-\epsilon/\epsilon_0$, so the $\SMP^*$ protocol is $\epsilon'=\epsilon/\epsilon_0$ correct, as needed. 
A similar argument establishes $\epsilon'$-correctness of the $\SMP^*$ protocol on inputs $(x,y)\in f^{-1}(1)$ which are $\epsilon$-correct in the $f$-BB84 protocol. 
Because this argument shows $f$-BB84 correctness on a given input implies $\SMP^*$ correctness on the same input, if the $f$-BB84 is $\epsilon$-correct on a fraction $1-\delta$ of inputs the $\SMP^*$ protocol is $\epsilon'$ correct on that fraction of inputs as well. 

The proof for $f$-routing is the same as the above, but now we replace \cref{lemma:fBB84emptyintersection} with \cref{lemma:fRemptyintersection}. 
In this setting, the referee now looks at the state $\rho_{RMM'}$ and determines if this is in $\tilde{S}_0^{\tilde{\epsilon}_0}$ or $\tilde{S}_{1}^{\tilde{\epsilon}_0}$. 
Otherwise, the proof is the same. 
\end{proof}

It is worth commenting on why the reduction from $f$-routing is to $\SMP^*$ rather than just $\SMP$. 
To understand this, notice that Alice and Bob cannot necessarily compute their gate choices directly from their inputs $x$ and $y$. 
Instead, they may use the outcomes of mid-circuit measurements to choose gates. 
To sample from these measurement outcomes, Alice and Bob need to share the same entangled state in their $\SMP$ protocol as is held in the $f$-routing protocol. 
A natural thought to avoid this is to have Alice and Bob purify their protocols, and apply only unitaries. 
In this case, however, classical processing used in the original protocol leads to additional quantum gates in the purified protocol. 
Thus, this would lower bound not the quantum gate complexity, but instead the total complexity including any classical part, and hence give a weaker bound. 

\subsection{Explicit lower bounds for \texorpdfstring{$f$}{TEXT}-BB84 and \texorpdfstring{$f$}{TEXT}-routing}

As a direct consequence of \cref{thm:reduction} (the reduction from $f$-routing or $f$-BB84 to $\SMP^*$), \cref{eq:smpandoneway} (which relates $\SMP^*$ and $C^*_{\epsilon}$) and \cref{eq:IPCCbound} (which lower bounds the communication complexity of the inner product function) we can state the following linear lower bound on the number of gates required to perform $f$-routing and $f$-BB84 tasks for $\IPfunc$.

\begin{corollary}
     Suppose $\IPfunc : \{0,1\}^n \times \{0,1\}^n \to \{0,1\}$ is the mod 2 inner product function, and $P$ is an $f$-BB84 protocol for $\IPfunc$ that is $\epsilon$-correct on a $1-\delta$ fraction of the inputs with $\delta+(1-\delta)\epsilon/\epsilon_0<1/2$ with $\epsilon_0=0.11$, uses $C_G(\IPfunc)$ gates drawn from a gate set of size $4$ and uses $C_M(\IPfunc)$ single qubit measurements in the computational basis.
    Then, 
    \begin{align}\label{eq:IPgatelowerboundexplicit84}
        (\log(q)+1)(2C_G(\IPfunc) +C_M(\IPfunc)) \geq \frac{1}{2}n + 2\log(1-2\epsilon') = \Omega(n) ,
    \end{align}
    where $q$ is the number of qubits held by Alice and Bob, and $\epsilon'= \delta +(1-\delta)\cdot \epsilon/\epsilon_0$. When $\delta=0$, we obtain the above lower bound when $\epsilon<0.055$.
\end{corollary}

\begin{corollary}
    Suppose $\IPfunc : \{0,1\}^n \times \{0,1\}^n \to \{0,1\}$ is the mod 2 inner product function, and $P$ is an $f$-routing protocol for $\IPfunc$ that is $\epsilon$-correct on a $1-\delta$ fraction of the inputs with $\delta+(1-\delta)\epsilon/\tilde{\epsilon}_0<1/2$ where $\tilde{\epsilon}_0=0.00085$, uses $C_G(\IPfunc)$ gates drawn from a gate set of size $4$ and uses $C_M(\IPfunc)$ single qubit measurements in the computational basis.
    Then, 
    \begin{align}\label{eq:IPgatelowerboundexplicitfR}
        (\log(q)+1)(2C_G(\IPfunc) +C_M(\IPfunc)) \geq \frac{1}{2}n + 2\log(1-2\epsilon') = \Omega(n) ,
    \end{align}
    where $q$ is the number of qubits held by Alice and Bob, and $\epsilon'= \delta +(1-\delta)\cdot \epsilon/\tilde{\epsilon}_0$.
    When $\delta=0$, we obtain the above lower bound when $\epsilon<0.00042$.
\end{corollary}

The reduction in \cref{thm:reduction} can easily be applied to other known $\SMP^*$ lower bounds for explicit functions. 
For instance, for the disjointness function we can use \cref{eq:disjointnessbound} and \cref{thm:reduction} to conclude the following corollary.

\begin{corollary}
     Suppose $\Disj : \{0,1\}^n \times \{0,1\}^n \to \{0,1\}$ is the disjointness function, and $P$ is an $f$-BB84 protocol for $\Disj$ that is $\epsilon<0.11$ correct on a $1-\delta$ fraction of the inputs with $\delta+(1-\delta)\epsilon/\epsilon_0<1/2$ where ${\epsilon}_0=0.11$, uses $C_G(\Disj)$ gates drawn from a gate set of size $4$ and also uses $C_M(\Disj)$ single qubit measurements in the computational basis.
    Then, 
    \begin{align}
        (\log(q)+1)(2C_G(\Disj) +C_M(\Disj)) = \Omega(\sqrt{n}),
    \end{align}
    where $q$ is the number of qubits held by Alice and Bob.
\end{corollary}

\begin{corollary}
     Suppose $\Disj : \{0,1\}^n \times \{0,1\}^n \to \{0,1\}$ is the disjointness function, and $P$ is an $f$-routing protocol for $\Disj$ that is $\epsilon$-correct on a $1-\delta$ fraction of the inputs where $\delta+(1-\delta)\epsilon/\tilde{\epsilon}_0<1/2$, $\tilde{\epsilon}_0=0.00085$, uses $C_G(\Disj)$ gates drawn from a gate set of size $4$ and also uses $C_M(\Disj)$ single qubit measurements in the computational basis.
    Then, 
    \begin{align}
        (\log(q)+1)(2C_G(\Disj) +C_M(\Disj)) = \Omega(\sqrt{n}) ,
    \end{align}
    where $q$ is the number of qubits held by Alice and Bob.
\end{corollary}

\subsection{Nearly matching upper bound for inner product}

For the inner product function, our lower bound on quantum operations is tight up to logarithmic factors. 
To see this, we use the garden-hose strategy \cite{buhrman2013garden} to give an upper bound. 
This strategy uses only Bell basis measurements and classical processing to attack any $f$-routing scheme. 
Adapted to the garden-hose setting, where we only have a single element in our gate set (the two qubit Bell basis measurement), our bound becomes
\begin{align}\label{eq:specializedbound}
    2(\log q+1)\, C_M = \Omega(n).
\end{align}
We construct a scheme using $q=O(n)$ EPR pairs and $O(n)$ measurements in the garden-hose model. 
This scheme then saturates the above bound up to logarithmic factors.

To construct the protocol, recall that in the garden-hose attack Alice and Bob share $N$ EPR pairs, and make Bell basis measurements connecting either the input system $Q$ to an EPR pair, or connecting two EPR pairs. 
A useful analogy is to a set of $N$ hoses, the ends of which can be connected together in pairs or connected to the tap, which plays the role of the input system. 
The water flowing through the pipes tracks where the system $Q$ ends up. 
A garden-hose attack for the inner product is as follows. 
Consider splitting $N$ into sets of 6 hoses. 
Alice will make measurements connecting her first set of 6 to her second set of 6. 
By wiring her connections appropriately, she can apply any permutation to the hoses. 
Define the permutations,
\begin{align}
    A &= (13)(24), \nonumber \\
    B &= (35)(46), \nonumber \\
    F &= (56). 
\end{align}
We then have Alice and Bob implement the permutations
\begin{align}
    S_i = A^{x_i}B^{y_i}FB^{y_i} A^{x_i}
\end{align}
by connecting subsequent layers of hoses with appropriate measurements. 
Alice connects the tap to the 1st hose from her first set and sends the 2nd hose from her final set to Bob. 
Notice that $S_i$ swaps the water from the first to the second hose iff $x_i\cdot y_i=1$, so that the water ends up on the first hose (and so the state ends up with Alice) if $\sum_i x_i y_i = 0$ mod $2$, and with Bob otherwise, as needed.

The total number of qubits used here is at most $12$ for each of the permutations $A^{x_i}, B^{y_i}$, $F$. 
There are a constant number of these per $S_i$, so the total number of qubits used $q$ is $O(n)$. 
Each qubit is measured at most once so the total number of measurements is also $O(n)$. 
Thus we find that
\begin{align}
    2(\log q+1)C_M =O(n\log n),
\end{align}
so this protocol shows \cref{eq:specializedbound} is tight up to a $\log n$ factor. 

\section{Discussion}

In this paper, we have given a new linear lower bound against the inner product function for $f$-routing and $f$-BB84. 
Our bound differs from earlier work in that it applies to a natural, low complexity function (inner product), and bounds the number of quantum operations necessary for an attack rather than the amount of shared entanglement. 
Assuming quantum gates are more difficult to implement than the same number of classical gates, our bound provides a separation in difficulty between an honest and dishonest player. 
Furthermore, it does so for a scheme that is computationally feasible for an honest prover.  

A key strategy in this paper compared to earlier ones has been to focus on the computational resources of the dishonest player. 
We can compare our bounds on quantum operations to bounds on system size, which are more common in the literature. 
To do so, consider restricting the dishonest player to circuits of depth $d$.
Then, their maximal number of gates is $C\sim d q$ for $q$ the number of qubits they control. 
From our bound then we obtain
\begin{align}
    q\log q \gtrsim \frac{n}{d} .
\end{align}
A natural restriction on the depth of the dishonest player's circuit is to take $d = O(\log (n))$, the same as (for the inner product function) the classical circuit depth of the honest player.\footnote{If the dishonest player has agents sitting at constant positions, and their gates take constant (in $n$) amount of time, this is enforced by relativity.} 
With this restriction, we obtain an almost linear lower bound on $q$, $q\log q\geq n/\log(n)\Rightarrow q\geq n/(\log(n))^2$.\footnote{We thank Philip Verduyn Lunel for making this point to us.}  
This compares favourably to the $q \gtrsim \log(n)$ lower bound obtained by \cite{bluhm2021position}, and furthermore avoids the issue of the purified view of \cite{bluhm2021position} counting what can in practice be classical systems towards the quantum system size.
Thus our bound, plus the assumption that the attacker's circuit depth should be similar to the honest player's, gives a stronger bound on quantum resource system size than has been proven previously. 
We can also note that applying deeper circuits to small systems is plausibly harder than shallow circuits on larger systems, so even relaxing this assumption and relying only on the gate lower bound seems a stronger bound than a linear bound on system size. 

\vspace{0.2cm}
\textbf{Note added:} After the publication of this article, the work \cite{bluhm2026complexity} shows an equivalence between $f$-routing and $f$-BB84. 
This means we can use our lower bound on $f$-BB84 and then use this equivalence to give a lower bound on $f$-routing with improved parameters. 

\vspace{0.2cm}
\noindent \textbf{Acknowledgements:} We thank Philip Verduyn Lunel for helpful discussions. 
AM is supported by the Perimeter Institute for Theoretical Physics. 
Research at Perimeter Institute is supported in part by the Government of Canada through the Department of Innovation, Science and Economic Development Canada and by the Province of Ontario through the Ministry of Colleges and Universities.

\bibliographystyle{unsrtnat}
\bibliography{biblio}

\end{document}